\documentclass[5p,preprint,10pt]{elsarticle}

\usepackage{amsmath,amsxtra,amssymb,latexsym, amscd, amsthm}
\usepackage{graphicx,color}
\usepackage{setspace}

\journal{...}

\newtheorem{theorem}{Theorem}[section]
\newtheorem{lemma}[theorem]{Lemma}

\newtheorem{definition}[theorem]{Definition}
\newtheorem{corollary}[theorem]{Corollary}
\newtheorem{example}[theorem]{Example}

\newtheorem{remark}[theorem]{Remark}

\begin{document}

\begin{frontmatter}



\title{A piecewise ellipsoidal reachable set estimation method for continuous bimodal piecewise affine systems}

\author[1]{Le Quang Thuan}
\ead{lequangthuan@qnu.edu.vn}

\cortext[cor1]{Corresponding author}
\author[1]{Phan Thanh Nam}
\ead{phanthanhnam@qnu.edu.vn}
\author[2,3]{Simone Baldi \corref{cor1}}
\ead{S.Baldi@tudelft.nl}

\address[1]{Department of Mathematics and Statistics, Quy Nhon University, 170 An Duong Vuong, Quy Nhon, Binh Dinh, Vietnam}
\address[2]{School of Mathematics, Southeast University, Jiulonghu Campus, Library 514, Nanjijng 211189, China}
\address[3]{Delft Center for Systems and Control, Delft University of Technology, The Netherlands}

\begin{abstract} In this work, the issue of estimation of reachable sets in continuous bimodal piecewise affine systems \textcolor{black}{is studied}. A new method is proposed, in the framework of ellipsoidal bounding, using piecewise quadratic Lyapunov functions. \textcolor{black}{Although bimodal piecewise affine systems can be seen as a special class of affine hybrid systems, reachability methods developed for affine hybrid systems might be inappropriately complex for bimodal dynamics. This work goes in the direction of exploiting the dynamical structure of the system to propose a simpler approach. More specifically,} because of the piecewise nature of the Lyapunov function, we first derive conditions to ensure that a given quadratic function is positive on half spaces. Then, we exploit the property of bimodal piecewise quadratic functions being continuous on a given hyperplane. Finally, linear matrix characterizations of the estimate of the reachable set \textcolor{black}{are derived.} 
\end{abstract}

\end{frontmatter}

\doublespacing

\section{Introduction} \label{sec:1}

Piecewise affine (PWA) systems form a special class of nonlinear dynamical systems, where the dynamics are indeed described by piecewise affine functions. A PWA system can be imagined as a collection of affine dynamics together with a partition of  the state space into polyhedral regions. Each region is associated with one particular affine dynamical system from the collection \cite{rodrigues19}. As the system states evolve, the dynamics switches if the state vector crosses from one polyhedral region to another \textcolor{black}{(state-dependent switching)}. PWA systems can model physical systems appearing in engineering such as relay systems, hysteresis systems and systems with saturation phenomena \cite{Thuan14}. In addition, \textcolor{black}{PWA systems} can be used to approximate nonlinear systems. Therefore, considerable attention has been paid to PWA systems and  important achievements have been published for classic control properties that include non-Zenoness, observability \cite{camlibel:06a} and controllability \cite{Kanat08, Thuan14,YURTSEVEN13},  stability and stabilization \cite{EREN14,IERVOLINO17}, $\mathcal L_2$-gain stability \cite{Waitman19}, {\em etc}.  

The reachability issue is one of the most important problems in the analysis of PWA systems. \textcolor{black}{Reachability analysis is also a major concern in the context of state estimation, see \cite{Durieu01} and reference therein.} Unfortunately, calculation of reachability sets is known as an undecidable issue in general piecewise systems \cite{Blondel99,Sontag95}. For linear systems, some special families of linear vector fields are proved with decidable reachability problem; see {\em e.g.} \textcolor{black}{\cite{Gan17, Lafferriere01,LEGUERNIC10}} and references therein. Several approaches have been proposed for calculating reachability sets of PWA systems: an overview of such approaches is given hereafter. The first approach is the research line proposed by Habets, Collins and van Schupen \cite{1643366}, \textcolor{black}{with} later developments in the works of Broucke and co-authors \cite{Broucke14, 1582905}. In \cite{1643366}, the authors studied the reach-avoid issue of  PWA hybrid systems on simplices: this terminology means that each affine dynamic is defined on a simplex and \textcolor{black}{inputs are} constrained in a polytope. The problem is to find an admissible piecewise control law that guarantees that every closed-loop hybrid state trajectory starting in an initial set can reach a target location after a finite number of discrete transitions while avoiding unsafe locations. Recently, \cite{Xiang18} adopted a similar philosophy, by considering the initial state inside a polytope and by approximating the reachable set via a multilayer neural network  with Rectified Linear Unit. \textcolor{black}{A second approach to reachability analysis of PWA systems aim to approximate outer reachable sets by means of polytopes, zonotopes with constraints, template polyhedra, barrier certificates\textcolor{black}{: some methods in this family, like} \cite{Alt10,Colas} estimate reachable sets on a finite time interval, and computational complexity inevitably increases for  longer time intervals \cite{Asarin}}\textcolor{black}{: other methods of the family consider unbounded time reachability approximation via non-linear optimization, convex optimization or lattice based fixed point computation \cite{Thao17,Thao11,Prajna04}. Similarly to the first family, the methods in this family also consider inputs constrained in a polytope}. A \textcolor{black}{third} approach to study the reachability of PWA systems in the absence of inputs is based on the tools of impact maps, proposed in \cite{HAMADEH08}. This approach  provides algorithms to estimate reachable sets on the switching surfaces from one switching surface rather than in the state spaces. Once the algorithm ends, it generates a series of upper and lower bound ellipsoidal subsets of the switching surfaces indicating which states the trajectories of the system can reach from a given initial set. A \textcolor{black}{fourth} approach to reachability of PWA systems was studied in \cite{Thuan14,1657409, YURTSEVEN13}. In these papers, \textcolor{black}{reachable set estimates are not studied; rather,} necessary and sufficient conditions are established to guarantee global reachability of  PWA bimodal systems (PWA systems composed of two dynamics) in the sense that for any two states in the state space there exists an input that steers the trajectory from one state to the other in finite time. \textcolor{black}{PWA bimodal systems are also the object of our work. Although the class \textcolor{black}{of} PWA bimodal systems is a special class of PWA systems, it covers some important classes of systems such as linear relay systems. As a matter of fact, this class of systems has attracted a
lot of interest from many researchers \cite{Kanat08,EREN14,880612,YURTSEVEN13}. }

 As a \textcolor{black}{fifth} and final family of approaches, ellipsoidal techniques based on suitable Lyapunov functions have been used for estimating reachable sets in various classes of systems. By using the quadratic Lyapunov function method combined with $S$-procedure, a linear matrix inequality (LMI) condition for ellipsoidal bounding of the reachable set was derived in \textcolor{black}{\cite{Boyd94,Pierre}} for linear systems. This idea is so flexible that it can be extended to singular systems \cite{Feng15}; to switched linear systems \textcolor{black}{\cite{Garoche,Baldi18, Xiang}} via multiple Lyapunov function approach; to time-delay linear systems   via the Lyapunov–Razumikhin method \cite{FRIDMAN03}, or via the Lyapunov – Krasovskii type functional \cite{Kim08}, or via the Lyapunov-Krasovskii functional together with  delay partition method \cite{Nam11}. \textcolor{black}{One characteristic of this method, which makes it different from the methods in the previous families is to consider inputs inside ellipsoids, such as minimum energy inputs. }

The work in this paper belongs to the \textcolor{black}{{fifth}} family of approaches, focusing specifically on continuous bimodal piecewise affine systems with exogenous disturbances. \textcolor{black}{Although PWA bimodal systems can be seen as a special class of PWA hybrid systems, reachability methods developed for general PWA hybrid might be inappropriately complex for bimodal dynamics. This work goes in the direction of exploiting the dynamical structure of the system to propose a simpler approach. For example, as compared to \cite{Broucke14,1643366,1582905}, we can consider reachability over infinite time intervals, with our method eventually resulting in a convex optimization problem that can be efficiently solved.}  As compared to the impact map approach \cite{HAMADEH08}, the estimation of reachable sets takes place in the state space rather than on the switching surface; we are not studying state estimation as in \cite{Durieu01}; similarly to \cite{Thuan14,YURTSEVEN13} we address the same class of systems (bimodal dynamics), but we are interested in a reachable set estimates rather than in conditions for reachability. For this class of systems, we derive new LMI characterizations for ellipsoidal bounding using a newly proposed piecewise approach. In order to develop this approach, we first need to derive conditions to ensure that a given quadratic function is non-negative on half spaces. Then, we exploit the property of bimodal piecewise quadratic functions that is continuous on a given hyperplane. Finally, we derive LMI conditions leading to a new (piecewise ellipsoidal) estimate of the reachable set. 

To conclude the overview, it is worth mentioning that LMI-based conditions for piecewise ellipsoidal estimation of reachable set  have appeared in literature \cite{8407489,NAKADA2004383,5159817}: although these methods can potentially address general (multimodal) PWA systems, they leave some open questions with respect to dealing with \textcolor{black}{state-dependent switching, possibly non-vanishing disturbances}, sliding-modes or with respect to making sure that the reachable set is connected. In particular, \textcolor{black}{\cite{NAKADA2004383} requires the disturbances to vanish at infinity and not to cause jumps of the state (in our method disturbances are bounded and possibly not vanishing); \cite{8407489} considers time-dependent switching (in place of state-dependent one); the backward reachable setting in \cite{5159817} is quite different from standard reachability analysis. In general, } proposed methods \textcolor{black}{based on piecewise quadratic Lyapunov functions} require special decompositions for the Lyapunov matrices over the different partitions, where some common matrices are assumed to exist: unfortunately, no general method is available for deriving such common matrices \textcolor{black}{(see discussion in Remark~\ref{rmk3.11})}. The proposed method does not require any special decomposition for the Lyapunov matrices composing the piecewise ellipsoid. In this sense, the proposed conditions (Lemmas 3.1, 3.2, 3.4 and Theorem 3.5) are novel and \textcolor{black}{are constructed in such a way to exploit the bimodal dynamics}. However, the proposed results are valid for the class of continuous bimodal piecewise affine systems: an extension to general multimodal PWA systems is not trivial and deserves future investigation. 


The paper is organized as follows. Section~\ref{sec:2} introduces bimodal piecewise affine  systems and some preliminaries. In Section \ref{sec:3}, we present the main results of the paper. The proposed approach is validated by a numerical example provided in Section~\ref{sec:4}. Finally, the paper ends with the conclusions in Section~\ref{sec:5}.

\noindent{\itshape{Notation}}: In this paper, we denote with $\mathbb R$ the set of all real numbers, $\mathbb R_+$ the set of all non-negative real numbers, and $\mathbb R^{n}_+$ the set of all $n$-tuple non-negative real numbers. The notation $\mathbb R^{n\times m}$ denotes the set of all real $n\times m$ matrices and the transpose of a real matrix $M\in \mathbb R^{n\times m}$ is denoted by $M^T$. The notation $\mathrm{He}(M)$ stands for the matrix $M + M^T$. We use the symbol $\ast$ in symmetric matrices to denote entries that follow from symmetry.  For a symmetric matrix $Q\in \mathbb R^{n\times n}$ and a linear subspace $\mathcal W$ of $\mathbb R^n$, we write $Q \overset{\mathcal W}{>} 0$ meaning that $x^T Q x >0$ for all nonzero $x \in \mathcal W$.  For a set $G$, $\mathrm{cl}(G)$ stands for its closure. The notation $\mathcal L_{1,\mathrm{loc}}(\mathbb R_+, \mathcal W)$ denotes the Lebesgue space of locally integrable functions from $\mathbb R_+$ to $\mathcal W$. For a nonempty subset $X$ of $\mathbb R^n$, its dual cone is denoted by $X^*$ and defined as $X^* =\{ z\in \mathbb R^n \ | \   z^Tx  \geqslant 0 \}.$  
\section{Bimodal piecewise affine systems} \label{sec:2}  

In this section, the class of bimodal piecewise affine systems is introduced 
\begin{equation}\label{BH2.1}
 	\dot x(t) =\begin{cases}
 	A_{1 } x(t) + B_1 w(t) +d_1   \text{ if } c^Tx(t) +f < 0 \\ A_{2} x(t) + B_2 w(t) +d_2  \text{ if } c^Tx(t) +f \geqslant 0
 	\end{cases}
\end{equation}
where  $x(t)\in \mathbb R^{n}$ is the state, $w(t) \in \mathbb R^m$ is the exogenous disturbance, $A_1, A_2 \in \mathbb R^{n \times n} $ and $B_1, B_2 \in \mathbb R^{n\times m}$ denote the state and input matrices for both modes, $d_1, d_2 \in \mathbb R^n $ and $c \in \mathbb R^n, f \in \mathbb R$ characterize the switching surface between the two modes. Assume that the disturbances $w(t)$ are taking their values in the following ellipsoid constraint set 
\begin{equation}\label{con:3}
\mathcal W :=\{w\in \mathbb R^m \ | \ w^T R_w w \leqslant 1 \}
\end{equation} 
where $R_w \in \mathbb R^{m\times m}$ is a given symmetric positive definite matrix.

In this paper, solution concept of system \eqref{BH2.1} is understood in the Carathéodory sense as follows.  
\begin{definition} [\cite{Kanat08}]\rm  A locally absolutely continuous function $x:\mathbb R_+ \to \mathbb R^n$ is said to be a solution of system \eqref{BH2.1} for the initial state $x_0$ and locally integrable disturbance $w(t)$ if $x(0) =x_0$ and $x$ satisfies the system \eqref{BH2.1} almost everywhere $t\in \mathbb R_+$. 
\end{definition}
When the system \eqref{BH2.1} is continuous, {\em i.e.} the following property holds 
\begin{equation}\label{c-3} 
c^T x +f = 0 \implies A_1 x + B_1w +d_1= A_2 x + B_2w+d_2,
\end{equation}
the right-hand side of system \eqref{BH2.1} is globally Lipschitz continuous (see {\em e.g.} \cite[Prop. 4.2.2]{pang:02}). Moreover, by \cite{Thuan14}, there exists $h \in \mathbb R^n$ such that 
\begin{equation}\label{c-4}
A_1 - A_2 = h c^T, \ B_1 - B_2 =h  0=0, d_1 -d_2 = hf,
\end{equation}
and we denote $B := B_1 =B_2$. In such cases, existence and uniqueness of solutions are guaranteed by the theory of ordinary differential equations. 

Let us denote the unique solution of system \eqref{BH2.1} for the initial state $x_0$ and disturbance $w$ by $x^w(t;x_0)$. 
\begin{definition}\rm  The reachable set (from the origin) of system \eqref{BH2.1} is the set of all states that can be reached in finite-time by starting from zero-state for any possible disturbances taking values in $\mathcal W$, {\em i.e. } 
$$
\mathcal R := \{x^w(t;0) \ | \ t\in \mathbb R_+, w \in \mathcal L_{1,\mathrm{loc}}(\mathbb R_+, \mathcal W) \}.
$$ 
\end{definition}
In the sequel, we consider system \eqref{BH2.1} with continuous right-hand side, {\em  i.e.} property \eqref{c-3} or \eqref{c-4} is fulfilled. We focus on the estimation of reachable set $\mathcal R$ of system \eqref{BH2.1} for any possible input disturbances in $\mathcal W$. In order to have a well-posed problem with bounded estimate, the matrices $A_1, A_2$ are assumed to be Hurwitz.

The idea behind estimating the reachable set is to consider a Lyapunov function $V(x)$ and find a region  outside which the derivative of the Lyapunov function is negative definite for any possible disturbances $w \in \mathcal W$. To do so, we employ the Lyapunov characterization that is similar to \cite{Boyd94,Feng15, Nam11}. Its proof is omitted for compactness and can be found in \cite{Boyd94, Feng15, Nam11}.
\begin{lemma} \label{lm:mai} Let $\alpha$ be a given positive scalar.  If a Lyapunov function $V(x)$ for system \eqref{BH2.1} exists satisfying  
	\begin{equation} \label{ine-1}
	\dot V(x(t)) +\alpha V(x(t)) - \alpha w^T(t)R_w w(t)\leqslant 0 
	\end{equation} 
	almost everywhere in $ t \in \mathbb R_+$ and along any trajectory of system \eqref{BH2.1}, then the reachable estimation set is contained in the 1-level set of the Lyapunov function
	$$ 
	\Gamma := \left \{ x \in \mathbb R^n \ | \ V(x) \leqslant 1 \right \}.
	$$ 
\end{lemma}

\section{Technical lemmas and main results} \label{sec:3} 

In this section, we propose a new method to estimate reachable sets of \textcolor{black}{bimodal} PWA systems in the framework of ellipsoidal bounding using piecewise quadratic Lyapunov functions. Because of the different nature of the Lyapunov function, we first need to establish auxiliary lemmas.  In the rest of this paper, without loss of generality, we may assume that $c^T = \begin{bmatrix} c_1& c_2 & \cdots &c_n \end{bmatrix}$ with $c_1 \ne 0$. Then, we define the $n\times (n-1)$-matrix $\hat R$ and the vector $r_0\in \mathbb R^n$ as 
\begin{equation} \label{hnm} 
\hat R:= \begin{bmatrix} -\dfrac{c_2}{c_1} &-\dfrac{c_3}{c_1} & \cdots & -\dfrac{c_n}{c_1} \\  1  & 0 &\cdots & 0 \\  0 & 1 & \cdots &0 \\ \vdots & \vdots & \cdots &\vdots \\ 0 &0 &\cdots & 1 \end{bmatrix}, r_0 :=  \begin{bmatrix} -\dfrac{f}{c_1}\\ 0 \\ \vdots \\ 0 \\ 0\end{bmatrix}. 
\end{equation} 
Also, we define
$ \Sigma_{-}  = \{ x \ | \ c^T x+ f <0\},   \Sigma_{0}  = \{ x \ | \ c^T x+ f =0\}, $ and $ \Sigma_{+}  = \{ x \ | \ c^T x+ f >0\}.
$ 

The first lemma presents  conditions to ensure that a given quadratic function is non-negative on a half space, {\em i.e.} the conditions under which the following inequality  holds: 
\begin{equation} \label{eq:lma}
 x^TP x + 2 b^T x + e > 0, \forall x \in (\Sigma_0 \cup \Sigma_+) \backslash \{0\}.
\end{equation} 

\begin{lemma}\label{lm-3.3} Suppose that $P\in \mathbb R^{n\times n}$ is a symmetric positive definite matrix. Then, the inequality \eqref{eq:lma} holds if and only if  
$$ 
\begin{bmatrix} r_0 &  c&  \hat R & -\hat R \\ 1 & 0& 0 &0 \end{bmatrix}^T  \begin{bmatrix}
P & b \\ b^T & e
\end{bmatrix} \begin{bmatrix} r_0 &  c&  \hat R & -\hat R \\ 1 & 0& 0 &0 \end{bmatrix}   \overset{\mathbb R_+^{2n}}{>} 0
$$
where $\hat R$ and $r_0$ are defined in \eqref{hnm}. 
\end{lemma}
\begin{proof} \textcolor{black}{Note that one can express the set $\Sigma_0 \cup \Sigma_+$  as  } 
\begin{align*}
\Sigma_0 \cup \Sigma_+ & = \{x \in \mathbb R^n \ | \ c^T x +f \geqslant 0 \} 
\\& = \left \{x \in \mathbb R^n \ | \ \begin{bmatrix} c^T & f \end{bmatrix} \begin{bmatrix} x\\ 1\end{bmatrix} \geqslant 0 \right \}. 
\end{align*}
\textcolor{black}{Next, we claim that
\begin{equation} \label{mnb}
\Sigma_0 \cup \Sigma_+ = \left \{ r_0+ \mu c + \hat R \theta_1 - \hat R \theta_2 \ \vline \ \begin{matrix} \mu \in \mathbb R_+, \\  \theta_1, \theta_2 \in \mathbb R^{n-1}_+ \end{matrix}  \right \}.
\end{equation}
Indeed, since $c^T r_0 = -f$, we have  
\begin{align*}
\begin{bmatrix} c^T & f \end{bmatrix} \begin{bmatrix} x\\ 1\end{bmatrix} &= \begin{bmatrix} c^T & f \end{bmatrix} \begin{bmatrix} r_0\\ 1\end{bmatrix} + \begin{bmatrix} c^T & f \end{bmatrix} \begin{bmatrix} x-r_0\\ 0\end{bmatrix} \\& = c^T(x-r_0).
\end{align*}
Therefore, $\begin{bmatrix} c^T & f \end{bmatrix} \begin{bmatrix} x\\ 1\end{bmatrix} \geqslant 0$ if and only if $c^T(x-r_0) \geqslant 0$. The latter holds if and only if 
$$
x-r_0 \in \mathrm{cone}(c) + \ker c^T 
$$
where $\mathrm{cone}(c)$ stands for the cone generated by $c$. Finally, it can be seen that 
$$
\ker c^T = \mathrm {im} \hat R = \big \{\hat R \theta_1 - \hat R \theta_2 \ | \  \theta_1, \theta_2 \in \mathbb R^{n-1}_+ \big \}. 
$$
Thus, the claim \eqref{mnb} is proven. }

\textcolor{black}{Now, due to \eqref{mnb},} if $x\in \Sigma_0 \cup \Sigma_+$, then 
$x= r_0  +  \mu c + \hat R\theta_1 + (-\hat R)\theta_2$ for some $\mu \in \mathbb R_+$ and $\theta_1, \theta_2 \in \mathbb R^{n-1}_+$. Thus 
\begin{equation} \label{eq:2w}
\begin{bmatrix} x \\ 1 \end{bmatrix} = \begin{bmatrix} r_0\\ 1 \end{bmatrix} + \mu \begin{bmatrix} c \\ 0 \end{bmatrix} + \begin{bmatrix} \hat R \\ 0 \end{bmatrix} \theta_1 + \begin{bmatrix} -\hat R \\ 0 \end{bmatrix} \theta_2.
\end{equation} 
In views of equality \eqref{eq:2w}, one can verify that  
$$
x^TP x + 2 b^T x + e = \begin{bmatrix} x\\ 1 \end{bmatrix}^T \begin{bmatrix} P & b\\ b^T & e \end{bmatrix} \begin{bmatrix} x\\ 1 \end{bmatrix} >  0 
$$
for all $x \in (\Sigma_0 \cup \Sigma_+) \backslash \{0\} $ if and only if 
$$ 
\begin{bmatrix} r_0 &  c&  \hat R & -\hat R \\ 1 & 0& 0 &0 \end{bmatrix}^T  \begin{bmatrix}
P & b \\ b^T & e
\end{bmatrix} \begin{bmatrix} r_0 &  c&  \hat R & -\hat R \\ 1 & 0& 0 &0 \end{bmatrix}    \overset{\mathbb R_+^{2n}}{>} 0. 
$$
The proof is complete. 
\end{proof}
The second lemma aims at providing conditions to make sure that a piecewise quadratic Lyapunov function is continuous on switching hyperplane $\Sigma_0$. 
\begin{lemma}\label{lm-3.4} The equality 
	$$x^TP_1 x + 2 b_1^T x + e_1  = x^TP_2 x + 2 b_2^T x +e_2, \forall x \in \Sigma_0
	$$
	holds if  
	$$
	\begin{bmatrix} r_0 &   \hat R & -\hat R \\ 1 & 0 &0 \end{bmatrix}^T \begin{bmatrix} P_1 -P_2 & b_1-b_2 \\ (b_1-b_2)^T & e_1-e_2 \end{bmatrix} \begin{bmatrix} r_0 &   \hat R & -\hat R \\ 1 & 0 &0 \end{bmatrix}  = 0. 
	$$
\end{lemma}
\begin{proof} One has  
\begin{align*}
\Sigma_0 & = \left \{x\in \mathbb R^n \ | \ c^T x +f =0\right \} \\& = \left \{ x\in \mathbb R^n \ | \ c^T x +f \geqslant 0 \right \}  \cap  \left \{ x\in \mathbb R^n \ | \ c^T x +f \leqslant 0 \right \} \\& = r_0 + \{\hat R \theta_1 + (-\hat R) \theta_2 \ | \  \theta_1, \theta_2 \in \mathbb R^{n-1}_+ \}.  
\end{align*}
If $x\in \Sigma_0$, then $ x= r_0 + \hat R\theta_1 - \hat R\theta_2$ for some $\theta_1, \theta_2 \in \mathbb R^{n-1}_+$. Thus, we have 
\begin{equation} \label{eq:2a}
\begin{bmatrix} x \\ 1 \end{bmatrix} = \begin{bmatrix} r_0 \\ 1 \end{bmatrix} + \begin{bmatrix} \hat R \\ 0 \end{bmatrix} \theta_1 - \begin{bmatrix} \hat R \\ 0 \end{bmatrix} \theta_2.
\end{equation} 
In views of equality \eqref{eq:2a}, one can see that
$$
\begin{bmatrix} x\\ 1 \end{bmatrix}^T\begin{bmatrix}  P_1-P_2 & b_1-b_2 \\ (b_1-b_2)^T & e_1 -e_2 \end{bmatrix}  \begin{bmatrix} x\\ 1 \end{bmatrix} =0, \forall x \in \Sigma_0
$$
if  
$$
\begin{bmatrix} r_0 &   \hat R & -\hat R \\ 1 & 0 &0 \end{bmatrix}^T \begin{bmatrix} P_1 -P_2 & b_1-b_2 \\ (b_1-b_2)^T & e_1-e_2 \end{bmatrix} \begin{bmatrix} r_0 &   \hat R & -\hat R \\ 1 & 0 &0 \end{bmatrix}  = 0. 
$$
The proof is complete. 
\end{proof}
Consider the set 
\begin{align*} 
\mathcal X & = \left \{\begin{bmatrix} x\\ w \end{bmatrix} \in \mathbb R^{n+m} \ \vline  \ \begin{matrix} c^Tx +f \geqslant 0,\\ 1 -w^T R_w w \geqslant 0 \end{matrix} \right  \}.
\end{align*} 
Let $M \in \mathbb R^{(n+m)\times (n+m)}$, $q \in \mathbb R^{n+m}$ and $p \in \mathbb R$. The last lemma provides a sufficient condition to guarantee that the following statement holds:
\begin{equation} \label{eq:st}
\begin{bmatrix} x\\ w \end{bmatrix}^T  M \begin{bmatrix} x\\ w \end{bmatrix} + 2 q^T \begin{bmatrix} x\\ w \end{bmatrix} + p \geqslant 0, \forall \begin{bmatrix} x\\ w \end{bmatrix} \in \mathcal X. 
\end{equation} 
Before presenting the last lemma, we need to introduce the concept of homogenization for a set.  
\begin{definition}[\cite{Sturm2003}] \label{def3.3}  For a non-empty set \textcolor{black}{ $D \subseteq \mathbb R^n$}, its homogenization is defined by 
$$
\mathcal H(D) =\mathrm{cl}\left \{ \begin{bmatrix} x\\ t \end{bmatrix} \in \mathbb R^{n+1} \ \vline \ x \in \mathbb R^n, t>0, x/t \in D \right  \}.
$$
\end{definition}
\textcolor{black}{
\begin{example} Let $D = [0,1] \subseteq \mathbb R$. Then, its homogenization is given by (see Figure 1)
$$
\mathcal H(D) = \left \{ \begin{bmatrix} x\\ t \end{bmatrix} \in \mathbb R^{2} \ \Big | \ t\geqslant 0, 0 \leqslant x \leqslant t \right \}.
$$	
\begin{figure}
\centering
	\includegraphics[scale=0.55]{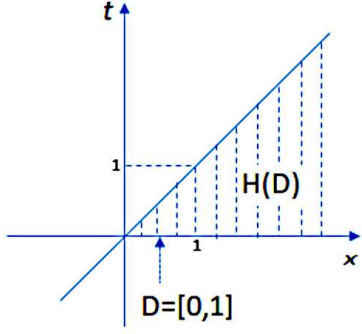}
	\caption{\textcolor{black}{The homogenization of D= [0,1]}}\label{fi1}
\end{figure}
\end{example}
}
\begin{lemma}\label{lm-3.6} The statement \eqref{eq:st} holds if and only if 
	$$
	\begin{bmatrix} M& q \\ q^T &p \end{bmatrix} \in \mathrm{cl} \left \{ Z    \ \vline  \ \begin{matrix} Z - \begin{bmatrix} 0 & 0 & c\sigma \\ 0& -\gamma R_w & 0  \\ \sigma c^T & 0& \gamma   \end{bmatrix}   \geqslant 0  \\ \text{ for some }  \gamma \geqslant 0, \sigma \geqslant 0 \end{matrix}  \right \}. 
	$$
\end{lemma}
\begin{proof} Define
\begin{align*} 
\mathcal D & = \left \{\begin{bmatrix} x\\ w\end{bmatrix}  \in \mathbb R^{n+m} \ | \ 1-w^T R_w w \geqslant 0 \right \} \\
& = \left \{\begin{bmatrix} x\\ w\end{bmatrix}  \in \mathbb R^{n+m} \ \vline \ 1 + \begin{bmatrix} x\\ w\end{bmatrix}^T \begin{bmatrix} 0&0 \\ 0 &-R_w  \end{bmatrix}  \begin{bmatrix} x\\ w\end{bmatrix} \geqslant 0 \right \} 
\\
& = \left \{\begin{bmatrix} x\\ w\end{bmatrix}  \in \mathbb R^{n+m} \ \vline \ \begin{bmatrix} x\\ w \\1 \end{bmatrix}^T \begin{bmatrix} 0&0 & 0\\ 0 &-R_w &0 \\ 0 &0 & 1 \end{bmatrix}  \begin{bmatrix} x\\ w \\ 1\end{bmatrix} \geqslant 0 \right \} 
\end{align*} 
and 
\begin{align*}
\mathcal L &= \left \{\begin{bmatrix} x\\ w \end{bmatrix} \in \mathbb R^{n+m} \ \vline \ c^Tx +f \geqslant 0 \right  \} 
\\& = \left \{\begin{bmatrix} x\\ w \end{bmatrix} \in \mathbb R^{n+m} \ \vline \ \begin{bmatrix} c^T & 0 & f \end{bmatrix} \begin{bmatrix} x\\ w \\ 1 \end{bmatrix} \geqslant 0 \right  \}  
\\& = \left \{\begin{bmatrix} x\\ w \end{bmatrix} \in \mathbb R^{n+m} \ \vline \ \begin{bmatrix} x^T & w^T &1 \end{bmatrix} \begin{bmatrix}   c \\ 0 \\ f \end{bmatrix} \geqslant 0 \right  \}.
\end{align*}
Then, $\mathcal X = \mathcal D \cap \mathcal L$. \textcolor{black}{Moreover, employing Definition~\ref{def3.3}, the homogenization of 
$$\mathcal D = \left \{\begin{bmatrix} x\\ w\end{bmatrix}  \in \mathbb R^{n+m} \ \vline \ 1-w^T R_w w \geqslant 0 \right \} $$ is 
\begin{align*}
\mathcal H(\mathcal D) & =\mathrm{cl}\left \{\begin{bmatrix} x\\ w \\ t \end{bmatrix}  \in \mathbb R^{n+m+1} \ \vline \ \begin{matrix} t>0 \\ \begin{bmatrix} x/t \\ w/t \end{bmatrix}  \in \mathcal D \end{matrix}  \right \} 
\\ & =\mathrm{cl}\left \{\begin{bmatrix} x\\ w \\ t \end{bmatrix}  \in \mathbb R^{n+m+1} \ \vline \ \begin{matrix} t>0 \\ 1-  (w/t)^T R_w (w/t) \geqslant 0 \end{matrix}  \right \} 
 \\ &
=\mathrm{cl}\left \{\begin{bmatrix} x\\ w \\ t \end{bmatrix}  \in \mathbb R^{n+m+1} \ \vline \ t^2  -w^T R_w w \geqslant 0, t > 0 \right \}.
\end{align*}}
It can be verified that the dual cone of  $\mathcal H(\mathcal D)$ is 
$$
\mathcal H(\mathcal D)^* = \left \{\begin{bmatrix} 0\\ 0 \\ \sigma \end{bmatrix}  \in \mathbb R^{n+m+1} \ \vline \ \sigma \in \mathbb R_+ \right \}.
$$
For any $\gamma, \sigma \in \mathbb R_+$, we have  
\begin{align*}\gamma \begin{bmatrix} 0&0 & 0\\ 0 &-R_w &0 \\ 0 &0 & 1 \end{bmatrix}  + \begin{bmatrix} c\\ 0 \\ f \end{bmatrix} \begin{bmatrix} 0&0 & \sigma \end{bmatrix} +  \begin{bmatrix} 0 \\ 0 \\ \sigma\end{bmatrix} \begin{bmatrix} c^T& 0&f\end{bmatrix} \\
= \begin{bmatrix} 0 & 0 & c\sigma \\ 0& -\gamma R_w & 0  \\ \sigma c^T & 0& \gamma +2f\sigma    \end{bmatrix}.
\end{align*}
By \cite[Theorem 3]{Sturm2003}, the statement \eqref{eq:st} holds if and only if 
$$\begin{bmatrix} M& q \\ q^T &p \end{bmatrix} \in \mathrm{cl} \left \{ Z    \ \vline  \ \begin{matrix} Z - \begin{bmatrix} 0 & 0 & c\sigma \\ 0& -\gamma R_w & 0  \\ \sigma c^T & 0& \gamma +2f\sigma   \end{bmatrix}   \geqslant 0  \\ \text{ for some }  \gamma, \sigma \in \mathbb R_+ \end{matrix}  \right \}.
$$
The proof is complete. 
\end{proof}
To continue, let us denote
\begin{subequations}\label{eq:lmk} 
	\begin{align} 
\tilde e_1 & = \begin{cases}
0 &, \text{ if } 0  \in \Sigma_0 \cup \Sigma_{-}  \\ e_1 \in \mathbb R &, \text{ if } 0  \in  \Sigma_+ ,
\end{cases}
\\ \tilde e_2 & = 
\begin{cases}
e_2 \in \mathbb R  &, \text{ if } 0   \in \Sigma_{-}  \\ 0 &, \text{ if } 0 \in \Sigma_0 \cup \Sigma_+.
\end{cases}
\end{align} 
\end{subequations} 
We are now ready to state and prove the main result of this paper.
 \begin{theorem}\label{mthrm-2} Consider the bimodal PWA system \eqref{BH2.1} subject to condition \eqref{con:3}.  Suppose that there exist scalars $\alpha > 0 $ and $\tilde e_1, \tilde  e_2  $ as in \eqref{eq:lmk}, positive definite symmetric matrices $P_1, P_2 \in \mathbb R^{n\times n}$ and vectors $b_1, b_2 \in \mathbb R^n$ such that the following conditions hold: 
 	\begin{equation}\label{eq:15}
 	\begin{bmatrix} r_0 &     R   \\ 1 & 0   \end{bmatrix}^T \begin{bmatrix} P_1 -P_2 & b_1-b_2 \\ (b_1-b_2)^T & \tilde e_1 - \tilde e_2 \end{bmatrix} \begin{bmatrix} r_0 &     R     \\ 1 & 0   \end{bmatrix}  = 0,
 	\end{equation} 
 	\begin{equation} \label{eq:16}
 	\begin{bmatrix} r_0 &  -c&   R  \\ 1 & 0& 0 \end{bmatrix}^T  
 	\begin{bmatrix}P_1 & b_1 \\ b_1^T & \tilde e_1 \end{bmatrix} 
 	\begin{bmatrix} r_0 &  -c&  R \\ 1 & 0& 0 \end{bmatrix}  \overset{\mathbb R_+^{2n}}{>} 0,
 	\end{equation}  
 	\begin{equation}\label{eq:17} 
 	\begin{bmatrix} r_0 &  c&    R  \\ 1 & 0& 0   \end{bmatrix}^T  \begin{bmatrix}
 	P_2 & b_2 \\ b_2^T & \tilde e_2 
 	\end{bmatrix} \begin{bmatrix} r_0 &  c&    R   \\ 1 & 0& 0   \end{bmatrix}  \overset{\mathbb R_+^{2n}}{>} 0,
 	\end{equation}   
 	\begin{multline}\label{eq:18}
 	\mathcal M_1 \in -\mathrm{cl} \left \{ Z    \ \vline  \ \begin{matrix} Z - \begin{bmatrix} 0 & 0 & -c\sigma \\ 0& -\gamma R_w & 0  \\ -\sigma c^T & 0& \gamma -2f\sigma   \end{bmatrix}   \geqslant 0  \\ \text{ for some }  \gamma, \sigma \in \mathbb R_+ \end{matrix}  \right \}
 	\end{multline}
 	and 
 	\begin{multline}\label{eq:19}
 	\mathcal M_2  \in - \mathrm{cl} \left \{ Z    \ \vline  \ \begin{matrix} Z - \begin{bmatrix} 0 & 0 & c\sigma \\ 0& -\gamma R_w & 0  \\ \sigma c^T & 0& \gamma  +2f\sigma  \end{bmatrix}   \geqslant 0  \\ \text{ for some }  \gamma,\sigma \in \mathbb R_+ \end{matrix}  \right \}
 	\end{multline}
 	where $R = \begin{bmatrix} \hat R & -\hat R \end{bmatrix}$ and
 	$$
 	\mathcal M_i = \begin{bmatrix} \mathrm{He}(A_i^T P_i) + \alpha P_i & * & * \\ B^T P_i & -\alpha R_w  & * \\ d_i^T P_i+ b_i^T A_i +\alpha b_i^T & b_i^T B & \alpha \tilde e_i +2b_i^T d_i  \end{bmatrix},  
 	$$ for $i=1,2$.  Then, the piecewise quadratic Lyapunov function 
 	\begin{equation} \label{Lya2}
 	V(x) = \begin{cases}
 	x^T P_1 x +2b_1^T x + \tilde e_1&,\text{if } c^T x +f < 0 \\ x^T P_2 x + 2b_2^T x + \tilde e_2&,\text{if } c^T x + f \geqslant 0 
 	\end{cases} 
 	\end{equation} 
 	 is continuous, strictly positive, radially unbounded and satisfies the inequality \eqref{ine-1}. As a result, the reachable set of system \eqref{BH2.1} is contained in the piecewise ellipsoidal set 
 	\begin{multline*} 
 	\Gamma_1 = \left \{ x \in \mathbb R^n \ \vline \ \begin{matrix} c^T x +f \leqslant 0\\  0 \leqslant \begin{bmatrix} x\\ 1\end{bmatrix}^T \begin{bmatrix}P_1 & b_1 \\ b_1^T &\tilde e_1 \end{bmatrix} \begin{bmatrix} x\\ 1\end{bmatrix}  \leqslant 1 \end{matrix} \right \} 
 	\\ \bigcup \left \{ x \in \mathbb R^n \ \vline \ \begin{matrix} c^T x +f \geqslant 0\\  0 \leqslant \begin{bmatrix} x\\ 1\end{bmatrix}^T \begin{bmatrix}P_2 & b_2 \\ b_2^T &\tilde e_2 \end{bmatrix} \begin{bmatrix} x\\ 1\end{bmatrix}  \leqslant 1 \end{matrix} \right \}.
 	\end{multline*} 
 \end{theorem} 
\begin{proof} Consider the Lyapunov function $V(x)$ defined by \eqref{Lya2}. First, due to \eqref{eq:15} and Lemma~\ref{lm-3.4}, it can be seen that $V$ is continuous on $\mathbb R^n$. Moreover, due to \eqref{eq:16},\eqref{eq:17}, Lemma~\ref{lm-3.3} and the determination of $\tilde e_1, \tilde e_2$ as in \eqref{eq:lmk}, we have $V(0) =0$, $V(x)>0$ for all $x\in \mathbb R^n, x \ne 0$ and $V(x)$ is radially unbounded. 
	
Next, we prove that $V$ satisfies inequality \eqref{ine-1}. Let $x(t)=x^w(t;0)$ be the trajectory of system \eqref{BH2.1} with a locally integrable disturbance $w(t)$. Then, for almost everywhere $t \in \mathbb R_+,$   $\dot x(t)$ exists satisfying
$$
\dot x(t) = A_i x(t) + B w(t) +d_i \text{ and } (-1)^{i+1} (c^T x(t) +f) \leqslant 0
$$
for some $i \in \{1,2\}$. We consider two possible cases: 

The first case is that $i=1$. In this case, we have $-c^T x(t) -f \geqslant 0$ and 
\begin{align*}
& \dot V(x(t)) +\alpha V(x(t)) - \alpha   w^T(t)R_w w(t)
 \\&  
 = \begin{bmatrix} x(t) \\ w(t) \end{bmatrix}^T  \begin{bmatrix} \mathrm{He}(A_1^T P_1) + \alpha P_1 & P_1B \\ B^TP_1 & -\alpha R_w  \end{bmatrix} \begin{bmatrix} x(t) \\ w(t) \end{bmatrix} 
 \\& 
 + 2\begin{bmatrix} d_1^T P_1+ b_1^T A_1 +\alpha b_1^T & b_1^T B \end{bmatrix}  \begin{bmatrix} x(t) \\ w(t) \end{bmatrix} 
+ \alpha \tilde e_1 + 2b_1^T d_1 
\\& 
=  \begin{bmatrix} x(t) \\ w(t) \\ 1 \end{bmatrix}^T \mathcal M_1 \begin{bmatrix} x(t) \\ w(t) \\ 1 \end{bmatrix}.
\end{align*}
Together with \eqref{eq:18} and Lemma~\ref{lm-3.6}, this presentation yields  
\begin{equation*}
\dot V(x(t)) +\alpha V(x(t)) - \alpha  w^T(t)R_w w(t) \leqslant 0.
\end{equation*} 
For the second case where $i=2$, we similarly have 
\begin{align*}
& \dot V(x(t)) +\alpha V(x(t)) - \alpha w^T(t) R_w w(t) 
\\&  
= \begin{bmatrix} x(t) \\ w(t) \\ 1 \end{bmatrix}^T  \mathcal M_2 \begin{bmatrix} x(t) \\ w(t) \\ 1 \end{bmatrix}.
\end{align*}
Due to \eqref{eq:19} and Lemma~\ref{lm-3.6}, we get   
\begin{equation*}
\dot V(x(t)) +\alpha V(x(t)) - \alpha  w^T(t)R_w w(t) \leqslant 0
\end{equation*} 
if the second case occurs.  In summary, for both cases, $V(x(t))$ satisfies inequality \eqref{ine-1}. The remain conclusion follows from Lemma~\ref{lm:mai} and the definition of $V(x)$ as \eqref{Lya2}.  
\end{proof}
Interestingly, conditions \eqref{eq:18} and \eqref{eq:19} can be greatly simplified. Note that if there exist $\gamma \in \mathbb R_+,\sigma \in \mathbb R_+$ such that 
\begin{multline} 
\begin{bmatrix} \mathrm{He}(A_1^T P_1) + \alpha P_1 & * & * \\ B^T P_1 & -\alpha R_w  & * \\ d_1^T P_1 + b_1^T A_1 +\alpha b_1^T & b_1^T B & \alpha \tilde e_1 +2b_1^T d_1  \end{bmatrix} \\ +  \begin{bmatrix} 0 & * & * \\ 0& -\gamma R_w & *  \\ -\sigma c^T & 0& \gamma -2f\sigma   \end{bmatrix}   \leqslant 0,
\end{multline}
then 
$$ 
\begin{bmatrix} \mathrm{He}(A_1^T P_1) + \alpha P_1 & * & * \\ B^T P_1 & -\alpha R_w  & * \\ d_1^T P_1 + b_1^T A_1 +\alpha b_1^T & b_1^T B & \alpha \tilde  e_1 +2b_1^T d_1 \end{bmatrix}
$$
belongs to the right-hand side of inclusion \eqref{eq:18}. It is similar manner for \eqref{eq:19}. Therefore, we get the following corollary from Theorem~\ref{mthrm-2}.
\begin{corollary}\label{cor:3.8} Consider the bimodal PWA system \eqref{BH2.1} subject to condition \eqref{con:3}.  Suppose that there exist positive definite symmetric matrices $P_1, P_2 \in \mathbb R^{n\times n}$, vectors $b_1, b_2 \in \mathbb R^n$ and scalars $\alpha >0, \gamma_1, \gamma_2,\sigma_1, \sigma_2 \in \mathbb R_+ $ and $\tilde e_1, \tilde e_2$ as in \eqref{eq:lmk} such that \eqref{eq:15}, \eqref{eq:16}, \eqref{eq:17} are satisfied and 
	\begin{equation} \label{eq:mx1}
	\begin{bmatrix}
	\mathrm{He}(A_1^T P_1) +\alpha P_1 & \ast  &   \ast   \\
	B^T P_1 & -(\gamma_1 +\alpha)R_w  & \ast \\ (P_1d_1+A_1^Tb_1+\alpha b_1-\sigma_1 c)^T & b_1^T B & \Delta_1 
	\end{bmatrix} \leqslant 0,  
	\end{equation}
	\begin{equation} \label{eq:mx2}
	\begin{bmatrix}
	\mathrm{He}(A_2^T P_2) +\alpha P_2 & \ast &   \ast  \\
	B^T P_2 &-(\gamma_2  +\alpha)R_w  & \ast \\ (P_2d_2+A_2^Tb_2+\alpha b_2+\sigma_2 c)^T & b_2^T B &\Delta_2
	\end{bmatrix} \leqslant 0,  
	\end{equation}
	where $\Delta_i =\alpha \tilde e_i +  2b_i^T d_i + \gamma_i +2(-1)^i f \sigma_i$. Then, the reachable set of system \eqref{BH2.1} is contained in the piecewise ellipsoidal set
	\begin{multline*} 
	\Gamma_2 = \left \{ x \in \mathbb R^n \ \vline \ \begin{matrix} c^T x +f \leqslant 0\\  0 \leqslant \begin{bmatrix} x\\ 1\end{bmatrix}^T \begin{bmatrix}P_1 & b_1 \\ b_1^T & \tilde e_1 \end{bmatrix} \begin{bmatrix} x\\ 1\end{bmatrix}  \leqslant 1 \end{matrix} \right \} 
	\\ \bigcup \left \{ x \in \mathbb R^n \ \vline \ \begin{matrix} c^T x +f \geqslant 0\\  0 \leqslant \begin{bmatrix} x\\ 1\end{bmatrix}^T \begin{bmatrix}P_2 & b_2 \\ b_2^T &\tilde e_2 \end{bmatrix} \begin{bmatrix} x\\ 1\end{bmatrix}  \leqslant 1 \end{matrix} \right \}.
	\end{multline*}
\end{corollary}
\begin{remark} \rm As common in ellipsoidal bounding methods, maximizing the trace of $P_1$ and $P_2$ subject to the proposed LMIs can be done, so as to make the estimated reachable set as small as possible. Note that, technically speaking, conditions \eqref{eq:18}, \eqref{eq:19}, \eqref{eq:mx1}, \eqref{eq:mx2} are not exactly LMIs due to the joint presence of $\alpha$ and $P_1, P_2$ in some block entries. However, this is a common feature to any ellipsoidal bounding method \cite{ Baldi18,Boyd94,Nam11}, where both the decreasing rate $\alpha$ and the Lyapunov matrices must be solved. This stems from condition \eqref{ine-1}, where $\alpha$ multiplies the Lyapunov function. A common method to solve this problem is to combine the LMI conditions with a one-dimensional search method over $\alpha$. 
\end{remark}

Some final remarks are made to clarify the contributions of the proposed method.

\begin{remark} \rm The main innovation of this work as compared to ellipsoidal bounding literature lies in the sequence of Lemmas~\ref{lm-3.3}, \ref{lm-3.4}, \ref{lm-3.6}, \textcolor{black}{designed in such a way to exploit the bimodal dynamics and} eventually leading to Theorem~\ref{mthrm-2}. Such lemmas have been introduced to handle positiveness of piecewise quadratic functions (Lemma~\ref{lm-3.3}) while exploiting the continuity properties of the system dynamics (Lemma~\ref{lm-3.4}) and its state-dependent switching (Lemma~\ref{lm-3.6}).
\end{remark}
\textcolor{black}{The previous remark allows us to elaborate on the difference between of the proposed approach and other LMI-based conditions for estimation of reachable set, \textcolor{black}{especially \cite{NAKADA2004383}}.} 

\begin{remark}\label{rmk3.11} \rm \textcolor{black}{In \cite{NAKADA2004383}, the Lyapunov functions are assumed in a certain decomposition form, {\em e.g.} $P_i = F_i^T T F_i$, where $T$ is common to all partitions and $F_i$ are selected \emph{a priori} to  satisfy  some certain continuity conditions. However, conditions provided in \cite[Theorem 1]{NAKADA2004383}  do not apply to bimodal piecewise affine system with possibly non-vanishing input as in \eqref{con:3}. In addition, if one employs the piecewise  Lyapunov function as $P_i = F_i^T T F_i$, the result on $T$ totally depends on the choice of matrices $F_1, F_2, f_1, f_2$ such that equality \cite[eq.(4)]{NAKADA2004383} holds. If we inappropriately choose these matrices, that may lead to conservative results. For example, for bimodal piecewise linear system
	$$
	\dot x(t) = \begin{cases}
	A_1 x(t) + B w(t) \text{ if } c^T x(t) \leqslant 0 \\ A_2 x(t) + B w(t) \text{ if } c^T x(t) \geqslant 0,
	\end{cases}
	$$
	if we take $F_1 = F_2 = c^T, f_1 = f_2 =0$ then \cite[eq.(4)]{NAKADA2004383} holds, but the corresponding piecewise Lyapunov function will boil down to a common Lyapunov function. So far, there are no results on how to select these matrices in some optimal/good way and derive common matrices such that the level set is bounded and connected. The conditions proposed in our approach avoid this special decomposition by carefully exploiting the continuity properties of the system dynamics and its state-dependent switching in Lemmas~\ref{lm-3.4} and \ref{lm-3.6}. } 
\end{remark}

\textcolor{black}{In view of the previous consideration on a common quadratic Lyapunov function as a special case of the decomposition approach, the} following remark can be stated. 

\begin{remark}\label{rm3.12} \rm  Suppose that there exist a positive scalar $\alpha>0$ and a positive definite  symmetric matrix $P>0$ such that 
	\begin{subequations} \label{eq:24}
		\begin{align}
		\begin{bmatrix} \mathrm{He}(A_1^T P) + \alpha P & PB \\ B^TP & -\alpha R_w   \end{bmatrix} & < 0, \\
		\begin{bmatrix} \mathrm{He}(A_2^T P) + \alpha P & PB \\ B^TP & -\alpha R_w   \end{bmatrix} & < 0.
		\end{align} 
	\end{subequations} 
	Then, the reachable set of system \eqref{BH2.1} is contained in the ellipsoidal set 
	$$ 
	\Gamma := \left \{ x \in \mathbb R^n \ | \ x^T P x  \leqslant 1 \right \}.
	$$ 
\end{remark}

\section{Numerical examples} \label{sec:4} 

\textcolor{black}{Because fair numerical comparisons with existing approaches based on piecewise quadratic Lyapunov functions are not possible due to quite different settings ({\em e.g.} time-dependent switching in \cite{8407489}, vanishing disturbances in \cite{NAKADA2004383}, {\em etc}), we will consider a common Lyapunov function \textcolor{black}{of Remark~\ref{rm3.12}} as a main mean of comparison.}
\begin{example}\rm Consider the bimodal piecewise affine system given as \eqref{BH2.1}, 
where $d_1 =d_2 =0, f=0$, $c^T = \begin{bmatrix} 0.5&0\end{bmatrix}$, $B_1^T = B_2^T = \begin{bmatrix} -0.5& 1 \end{bmatrix}$ and 
$$ A_1 = \begin{bmatrix} -0.5& -0.4 \\ 3 & -0.5 \end{bmatrix}, \ A_2 = \begin{bmatrix} -3& -0.4 \\ -0.5 & -0.5\end{bmatrix}. $$
This is a PWA system in which the phase planes of the two subsystems are almost "orthogonal" to each other, as we will see in a while. The disturbance $w$ is constrained in the interval $[-1,1]$, {\em i.e.} $R_w =1$.  In this example, we employ both quadratic and piecewise quadratic Lyapunov functions for estimating its reachable set. By maximizing the traces of involved Lyapunov matrices and taking $\alpha =0.4$, the results are shown in Figure 2. In this figure, the  green solid line presents reachable set ellipsoidal bounding when common quadratic Lyapunov function is employed; the combined red- and blue solid lines present reachable set piecewise ellipsoidal bounding when piecewise quadratic Lyapunov function is employed. Inside the estimated sets we can see the result of $1000$ trajectories obtained with $1000$ realizations of a random disturbance $w$ uniformly distributed in $[-1,1]$. It can be observed that using piecewise quadratic Lyapunov function gives a better approximation to the reachable set than the common quadratic Lyapunov function.  In particular, the piecewise quadratic Lyapunov function is able to capture the "orthogonal" nature of the phase planes. \textcolor{black}{The computational time of the common Lyapunov function is 0.0523\textcolor{black}{s} (yalmiptime) + 0.0527\textcolor{black}{s} (solvertime), whereas the computational time of the piecewise Lyapunov function is 0.0461\textcolor{black}{s} (yalmiptime) + 0.0549\textcolor{black}{s} (solvertime). All simulations are performed in Matlab R2019b with Yalmip/Sedumi, on a laptop with Intel Core i7-8750H, RAM 32.0 GB.}
\end{example}
\begin{figure}
\centering
	\includegraphics[scale=0.6]{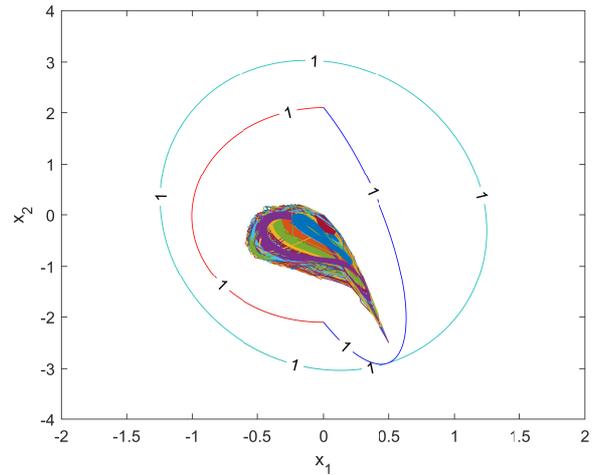}
	\caption{Estimates of the reachable set: common quadratic Lyapunov function (cyan solid line) vs. piecewise quadratic Lyapunov function (red-blue solid line). Inside the
		estimated sets we can see the result of 1000 trajectories
		obtained with 1000 realizations of a random bounded disturbance.}
\end{figure}

\begin{example} \rm \textcolor{black}{We consider the mechanical system shown in Figure~3, which is a popular benchmark in the bimodal systems community \cite{Kanat08}. 	
	The state comprises $x_1$ and $x_2$, the displacements of the left and right cart from the tip of the leftmost spring, and $x_3$ and $x_4$, corresponding velocities. Let us denote the masses of the carts by $m_1$ and $m_2$, the spring constants by $k'$ and $k$, and the damping constant by $d$. The equations of motion can be derived as 
\begin{equation}\label{lb23}
\dot x(t) = \begin{cases}
A_1 x(t) + B(F(t)+w(t)) \text{ if } c^T x(t) \leqslant 0 \\ A_2 x(t) + B(F(t)+w(t)) \text{ if } c^T x(t) \geqslant 0
\end{cases}
\end{equation}
where $x^T= \begin{bmatrix} x_1&x_2&x_3&x_4\end{bmatrix}$, $F$ is the force applied to the right cart, $w$ is a bounded disturbance, and $c^T = \begin{bmatrix} 1&0&0&0 \end{bmatrix},$ $B^T = \begin{bmatrix} 0&0&0&1\end{bmatrix},$ 
\begin{align*}
A_1 &= \begin{bmatrix} 0&0&1&0\\ 0&0&0&1\\ \dfrac{-(k+k')}{m_1}& \dfrac{k}{m_1}&\dfrac{-d}{m_1}&\dfrac{d}{m_1} \\ \dfrac{-k}{m_2}& \dfrac{k}{m_2}& \dfrac{-d}{m_2}&\dfrac{d}{m_2}\end{bmatrix}, \\ 
A_2 &= \begin{bmatrix} 0&0&1&0\\ 0&0&0&1\\ \dfrac{-k}{m_1}& \dfrac{k}{m_1}&\dfrac{-d}{m_1}&\dfrac{d}{m_1} \\ \dfrac{-k}{m_2}& \dfrac{k}{m_2}& \dfrac{-d}{m_2}&\dfrac{d}{m_2}\end{bmatrix}.
\end{align*}
Note that the two conditions in \eqref{lb23} represent the impact with the spring. In order to obtain a bounded reachability set, we add a stabilizing controller 
\begin{equation}
F(t) = -K^Tx(t) 
\end{equation}
where $K \in \mathbb R^4$ is a proportional-derivative control gain. The experiments are performed with
$
K^T= \begin{bmatrix}10&10&1&1 \end{bmatrix} 
$ and 
$
m_1 = 1, m_2 = 2, k = 0.1, k' = 0.1, d = 0.1, R_w =1, \alpha = 0.1.
$
\begin{figure}\label{fig3}
	\centering
		\includegraphics[scale=0.6]{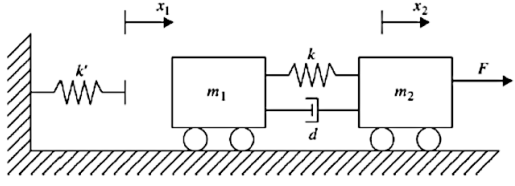}
		\caption{Linear mechanical system with a one-sided spring.}
	\end{figure}
The application of Remark~\ref{rm3.12} results in a common Lyapunov function with positive definite matrix 
$$
P =\begin{bmatrix}
3.0035 & 1.5178 & -0.0157 & 0.0227 \\
1.5178 & 1.5269 & -0.6213 & 0.0018 \\
 -0.0157 & -0.6213 & 6.2020 & 0.2085 \\
 0.0227 & 0.0018 & 0.2085 & 0.1492
 \end{bmatrix}  
$$
whereas the application of Corollary~\ref{cor:3.8} results in a piecewise Lyapunov function with positive definite matrices
$$
P_1 =\begin{bmatrix}
12.3671 &   1.5443 &  -0.0777  &  0.1213 \\
    1.5443  &  1.9862  & -4.8429  &  -0.0356 \\
    -0.0777 &  -4.8429 &  45.0481 &   0.5552 \\
    0.1213  & -0.0356  &  0.5552  &  0.1534 
\end{bmatrix},  
$$
$$
P_2 =\begin{bmatrix}
11.3517  &  1.6612 &  -1.0981  &  0.1041 \\
     1.6612  &  1.9862  & -4.8429  & -0.0356 \\
    -1.0981 &  -4.8429  &  45.0481  &  0.5552 \\
    0.1041  & -0.0356 &   0.5552  &  0.1534 
\end{bmatrix}  
$$
and $b_1,b_2 \approx 0$ (of the order of $10^{-6}$).
Due to the presence of four states, it is not possible to plot the 1-level set of common/piecewise Lyapunov functions. However, \textcolor{black}{we have projected the reachable set in the bidimensional subspace $(x_1,x_3)$, so that the reachanle set can be visualized in the plane as in Fig.~\ref{h1r1}. In addition,} the interested reader can easily verify that $P_1>P$ and $P_2>P$, which implies that the 1-level set of the piecewise Lyapunov function is smaller than the  1-level set of the common Lyapunov function. This further verifies that the proposed approach leads to a smaller estimate of the reachable set. \textcolor{black}{For comparison purposes, we have used the same numerical example via the Matlab-based reachability toolbox COntinuous Reachability Analyzer (CORA), available at https://github.com/TUMcps/CORA). With this software, the reachable set is approximated as a zonotope, and we have tried a zonotope of order 40 (larger dash-dotted line in Fig.~\ref{h1r1}) and a zonotope of order 50 (smaller dash-dotted line in Fig.~\ref{h1r1}). For the zonotope of order 40, we need more than 200s to solve the problem, whereas we need more than 215s for the zonotope of order 50 (using the same computer platform as before). Increasing the order of the zonotope increases the computational time, but does not improve the estimate of the reachable set.} Let us finally report that the computational time of the common Lyapunov function is 0.0547\textcolor{black}{s} (yalmiptime) + 0.0873\textcolor{black}{s} (solvertime), whereas the computational time of the piecewise Lyapunov function is 0.0507\textcolor{black}{s} (yalmiptime) + 0.1913\textcolor{black}{s} (solvertime). \textcolor{black}{Therefore, for a 4th order system, our method solves the problem in 0.1s, against 200s of CORA (3 orders of magnitude difference).}
}
\end{example}
\begin{figure}[h]
	\centering
	\includegraphics[scale=0.6]{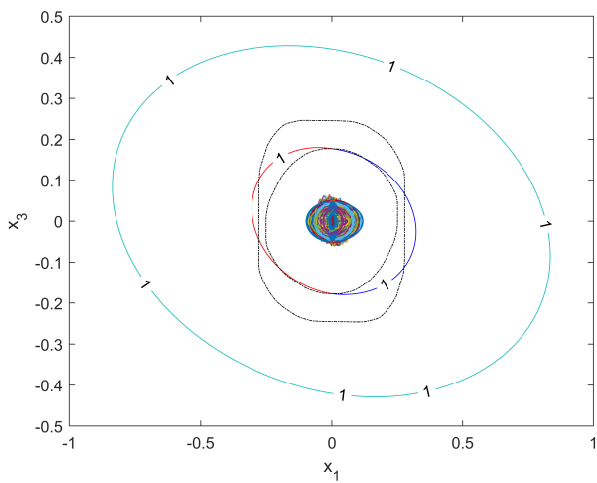}
	\caption{Estimates of the reachability set with common Lyapunov function (cyan solid line, computational time $\approx$ 0.1s); with proposed bimodal method (blue-red solid line, computational time $\approx$ 0.1s); with CORA with 40th-order zonotope (larger black dash-dotted line, computational time $\approx$ 200s); with CORA with 50th-order zonotope (smaller black dash-dotted line, computational time $\approx$ 215s). Increasing the order of the zonotope increases the computational time, but does not improve the result.}\label{h1r1}
\end{figure}
\section{Conclusions}\label{sec:5} 

In this paper, we have proposed a new method, in the framework of ellipsoidal bounding, to estimate the reachable sets of continuous bimodal piecewise affine systems using piecewise quadratic Lyapunov functions. Due to the different nature of the Lyapunov function, we had to derive conditions to ensure that a given quadratic function is positive on half spaces. Then, we exploited the property of bimodal piecewise quadratic functions being continuous on a given hyperplane. Finally, we derived linear matrix inequality characterizations of the estimate of the reachable set. 

A relevant future work is to study the extension of the proposed method for trimodal or multimodal systems, {\em e.g. } by exploiting special representations of system dynamics, such as hinging hyperplanes \cite{Moustakis18, Xu09}. 

\section*{Acknowledgments}

\textcolor{black}{The research of the first two authors were supported by Vingroup Innovation Foundation under Grant VINIF.2019. DA09. The research of the last author was partly supported by Double Innovation Plan under Grant 4207012004, and by Special Funding for Overseas talents under Grant 6207011901.}

\bibliographystyle{plain} 
\bibliography{refs_file_1}

\begin{thebibliography}{10}

\bibitem{Thao17}
{A. Arvind} and {D.Thao}.
\newblock Augmented complex zonotopes for computing invariants of affine hybrid
  systems.
\newblock In {\em Abate A., Geeraerts G. (eds) Formal Modeling and Analysis of
  Timed Systems. FORMATS 2017. Lecture Notes in Computer Science, vol 10419.},
  pages 97--115. Springer, 2017.

\bibitem{Garoche}
A.~{Adjé} and PL. {Garoche}.
\newblock Automatic synthesis of piecewise linear quadratic invariants for
  programs.
\newblock In D.~{D’Souza}, A.~{Lal}, and K.G. {Larsen}, editors, {\em
  International Workshop on Verification, Model Checking, and Abstract
  Interpretation 2015}, volume 8931 of {\em Lecture Notes in Computer Science},
  pages 99--116. Springer, Berlin, Heidelberg, 2015.

\bibitem{Alt10}
M.~{Althoff}, O.~{Stursberg}, and M.~{Buss}.
\newblock Computing reachable sets of hybrid systems using a combination of
  zonotopes and polytopes.
\newblock {\em Nonlinear Analysis: Hybrid Systems}, 4(2):233 -- 249, 2010.

\bibitem{Asarin}
E.~{Asarin}, T.~{Dang}, G.~{Frehse}, A.~{Girard}, C.~{Le Guernic}, and
  O.~{Maler}.
\newblock Recent progress in continuous and hybrid reachability analysis.
\newblock In {\em 2006 IEEE Conference on Computer Aided Control System Design,
  2006 IEEE International Conference on Control Applications, 2006 IEEE
  International Symposium on Intelligent Control}, pages 1582--1587, 2006.

\bibitem{Baldi18}
S.~{Baldi} and W.~{Xiang}.
\newblock Reachable set estimation for switched linear systems with dwell-time
  switching.
\newblock {\em Nonlinear Analysis: Hybrid Systems}, 29:20 --33, 2018.

\bibitem{Blondel99}
V.~D. {Blondel} and J.~N. {Tsitsiklis}.
\newblock Complexity of stability and controllability of elementary hybrid
  systems.
\newblock {\em Automatica}, 35(3):479 -- 489, 1999.

\bibitem{Boyd94}
S.~{Boyd}, L.~El {Ghaoui}, E.~{Feron}, and V.~{Balakrishnan}.
\newblock {\em Linear Matrix Inequalities in System and Control Theory}.
\newblock Volume 15 of Studies in Applied Mathematics. Society for Industrial
  and Applied Mathematics, 1994.

\bibitem{Broucke14}
M.~E. {Broucke} and M.~{Ganness}.
\newblock Reach control on simplices by piecewise affine feedback.
\newblock {\em SIAM Journal on Control and Optimization}, 52(5):3261--3286,
  2014.

\bibitem{camlibel:06a}
M.~K. {Camlibel}, J-S. {Pang}, and J.~{Shen}.
\newblock Conewise linear systems: non-{Z}enoness and observability.
\newblock {\em SIAM Journal on Control and Optimization}, 45(5):1769--1800,
  2006.

\bibitem{Kanat08}
M.K. {Camlibel}, W.P.M.H. {Heemels}, and J.M. {Schumacher}.
\newblock A full characterization of stabilizability of bimodal piecewise
  linear systems with scalar inputs.
\newblock {\em Automatica}, 44(5):1261 -- 1267, 2008.

\bibitem{Thao11}
{D. Thao} and {T.M. Gawlitza}.
\newblock Template-based unbounded time verification of affine hybrid automata.
\newblock In {\em Programming Languages and Systems}, pages 34--49. Springer
  Berlin Heidelberg, 2011.

\bibitem{Durieu01}
C.~{Durieu}, É. {Walter}, and B.~{Polyak}.
\newblock Multi-input multi-output ellipsoidal state bounding.
\newblock {\em Journal of Optimization Theory and Applications}, 111:273--303,
  2001.

\bibitem{EREN14}
Y.~{Eren}, J.~{Shen}, and M.~K. {Camlibel}.
\newblock Quadratic stability and stabilization of bimodal piecewise linear
  systems.
\newblock {\em Automatica}, 50(5):1444 -- 1450, 2014.

\bibitem{pang:02}
F.~{Facchinei} and J-S. {Pang}.
\newblock {\em Finite {D}imensional {V}ariational {I}nequalities and
  {C}omplementarity {P}roblems}, volume~1.
\newblock Springer, New York, 2003.

\bibitem{Feng15}
Z.~{Feng} and J.~{Lam}.
\newblock On reachable set estimation of singular systems.
\newblock {\em Automatica}, 52:146 -- 153, 2015.

\bibitem{FRIDMAN03}
E.~{Fridman} and U.~{Shaked}.
\newblock On reachable sets for linear systems with delay and bounded peak
  inputs.
\newblock {\em Automatica}, 39(11):2005 -- 2010, 2003.

\bibitem{Gan17}
T.~{Gan}, M.~{Chen}, Y.~{Li}, B.~{Xia}, and N.~{Zhan}.
\newblock Reachability analysis for solvable dynamical systems.
\newblock {\em IEEE Transactions on Automatic Control}, 63(7):2003--2018, 2018.

\bibitem{1643366}
L.~C. G. J.~M. {Habets}, P.~J. {Collins}, and J.~H. {van Schuppen}.
\newblock Reachability and control synthesis for piecewise-affine hybrid
  systems on simplices.
\newblock {\em IEEE Transactions on Automatic Control}, 51(6):938--948, 2006.

\bibitem{HAMADEH08}
A.~{Hamadeh} and J.~{Goncalves}.
\newblock Reachability analysis of continuous-time piecewise affine systems.
\newblock {\em Automatica}, 44(12):3189 -- 3194, 2008.

\bibitem{IERVOLINO17}
R.~{Iervolino}, D.~{Tangredi}, and F.~{Vasca}.
\newblock Lyapunov stability for piecewise affine systems via
  cone-copositivity.
\newblock {\em Automatica}, 81:22 -- 29, 2017.

\bibitem{880612}
J.~{Imura} and A.~{van der Schaft}.
\newblock Characterization of well-posedness of piecewise-linear systems.
\newblock {\em IEEE Transactions on Automatic Control}, 45(9):1600--1619, 2000.

\bibitem{Kim08}
J-H. {Kim}.
\newblock Improved ellipsoidal bound of reachable sets for time-delayed linear
  systems with disturbances.
\newblock {\em Automatica}, 44(11):2940 -- 2943, 2008.

\bibitem{Lafferriere01}
G.~{Lafferriere}, G.~J. {Pappas}, and S.~{Yovine}.
\newblock Symbolic reachability computation for families of linear vector
  fields.
\newblock {\em Journal of Symbolic Computation}, 32(3):231 -- 253, 2001.

\bibitem{Colas}
C.~{Le Guernic}.
\newblock {\em Reachability Analysis of Hybrid Systems with Linear Continuous
  Dynamics}.
\newblock PhD thesis, Universit\'e Joseph-Fourier - Grenoble I, English.
  tel-00422569v2., 2009.

\bibitem{LEGUERNIC10}
C.~{Le Guernic} and A.~{Girard}.
\newblock Reachability analysis of linear systems using support functions.
\newblock {\em Nonlinear Analysis: Hybrid Systems}, 4(2):250 -- 262, 2010.

\bibitem{8407489}
J.~{Li}, J.~{Shi}, and Z.~{Feng}.
\newblock Reachable set estimation for discrete-time {T}-{S} fuzzy singular
  systems based on piecewise {L}yapunov function.
\newblock In {\em 2018 Chinese Control And Decision Conference (CCDC)}, pages
  2189--2193, 2018.

\bibitem{Moustakis18}
N.~{Moustakis}, B.~{Zhou}, L.Q. {Thuan}, and S.~{Baldi}.
\newblock Fault detection and identification for a class of continuous
  piecewise affine systems with unknown subsystems and partitions.
\newblock {\em International Journal of Adaptive Control and Signal
  Processing}, 32(7):980--993, 2018.

\bibitem{NAKADA2004383}
H.~{Nakada} and K.~{Takaba}.
\newblock Reachable set analysis of uncertain piecewise affine systems and its
  application to {T}-{S} fuzzy systems.
\newblock {\em IFAC Proceedings Volumes}, 37(11):383 -- 388, 2004.

\bibitem{Nam11}
P.~T. {Nam} and P.~N. {Pathirana}.
\newblock Further result on reachable set bounding for linear uncertain
  polytopic systems with interval time-varying delays.
\newblock {\em Automatica}, 47(8):1838 -- 1841, 2011.

\bibitem{rodrigues19}
L.~{Rodrigues}, B.~{Samadi}, and M.~{Moarref}.
\newblock {\em Piecewise Affine Control: Continuous-Time, Sampled-Data, and
  Networked Systems}.
\newblock Advances in Design and Control. Society for Industrial and Applied
  Mathematics, 2019.

\bibitem{1582905}
B.~{Roszak} and M.~E. {Broucke}.
\newblock Necessary and sufficient conditions for reachability on a simplex.
\newblock In {\em Proceedings of the 44th IEEE Conference on Decision and
  Control}, pages 4706--4711, 2005.

\bibitem{Pierre}
P.~{Roux}, R.~{Jobredeaux}, P-L. {Garoche}, and E.~{F\'{e}ron}.
\newblock A generic ellipsoid abstract domain for linear time invariant
  systems.
\newblock In {\em Proceedings of the 15th ACM International Conference on
  Hybrid Systems: Computation and Control}, HSCC '12, page 105–114, New York,
  NY, USA, 2012. Association for Computing Machinery.

\bibitem{Prajna04}
{S. Prajna} and {A. Jadbabaie}.
\newblock Safety verification of hybrid systems using barrier certificates.
\newblock In {\em Hybrid Systems: Computation and Control}, pages 477--492.
  Springer, 2004.

\bibitem{Sontag95}
E.~{Sontag}.
\newblock From linear to nonlinear: some complexity comparisons.
\newblock In {\em Proceedings of 1995 34th IEEE Conference on Decision and
  Control}, volume~3, pages 2916--2920 vol.3, 1995.

\bibitem{Sturm2003}
J.~F. {Sturm} and S.~{Zhang}.
\newblock On cones of nonnegative quadratic functions.
\newblock {\em Mathematics of Operations Research}, 28(2):246--267, 2003.

\bibitem{Thuan14}
L.~Q. {Thuan} and M.~K. {Camlibel}.
\newblock Controllability and stabilizability of a class of continuous
  piecewise affine dynamical systems.
\newblock {\em SIAM Journal on Control and Optimization}, 52(3):1914--1934,
  2014.

\bibitem{Waitman19}
S.~{Waitman}, P.~{Massioni}, L.~{Bako}, and G.~{Scorletti}.
\newblock Incremental {$L_2$}-gain stability of piecewise-affine systems with
  piecewise-polynomial storage functions.
\newblock {\em Automatica}, 107:224 -- 230, 2019.

\bibitem{5159817}
M.~{Wu}, G.~{Yan}, Z.~{Lin}, and M.~{Liu}.
\newblock Characterization of backward reachable set and positive invariant set
  in polytopes.
\newblock In {\em 2009 American Control Conference}, pages 4351--4356, 2009.

\bibitem{Xiang}
W.~{Xiang}, H.~{Tran}, and T.~T. {Johnson}.
\newblock Output reachable set estimation for switched linear systems and its
  application in safety verification.
\newblock {\em IEEE Transactions on Automatic Control}, 62(10):5380--5387,
  2017.

\bibitem{Xiang18}
W.~{Xiang}, H.~{Tran}, J.~A. {Rosenfeld}, and T.~T. {Johnson}.
\newblock Reachable set estimation and safety verification for piecewise linear
  systems with neural network controllers.
\newblock In {\em 2018 Annual American Control Conference (ACC)}, pages
  1574--1579, 2018.

\bibitem{Xu09}
J.~{Xu}, X.~{Huang}, and S.~{Wang}.
\newblock Adaptive hinging hyperplanes and its applications in dynamic system
  identification.
\newblock {\em Automatica}, 45(10):2325 -- 2332, 2009.

\bibitem{1657409}
J.~{Xu} and L.~{Xie}.
\newblock Controllability and reachability of discrete-time planar bimodal
  piecewise linear systems.
\newblock In {\em 2006 American Control Conference}, pages 6 pp.--, 2006.

\bibitem{YURTSEVEN13}
E.~{Yurtseven and M.K. {Camlibel} and W.P.M.H. {Heemels}}.
\newblock Controllability of a class of bimodal discrete-time piecewise linear
  systems.
\newblock {\em Systems \& Control Letters}, 62(4):338 -- 344, 2013.

\end{thebibliography}

\end{document}